\documentclass{amsart}

\usepackage{amssymb}
\usepackage{amsfonts}
\usepackage{graphicx}

\newtheorem{theorem}{Theorem}
\theoremstyle{plain}

\newtheorem{definition}{Definition}
\newtheorem{property}{Property}
\newtheorem{lemma}{Lemma}
\numberwithin{equation}{section}

\begin{document}
\title[Code uniqueness of 4-ccd clusters in SCM]{Code uniqueness of 4-cell connected disjoint clusters in spatial chromatic model}
\author{Weining Zhu}
\address[ZHU, W. N.]{Laboratory of Spatial Computation and Analysis, Zhejiang University}
\email{zhuwn@zju.edu.cn}
\thanks{This work was supported by the Natural Science Foundation of Zhejiang Province (No. LY17D010005)}
\date{July 17, 2026}

\begin{abstract}
The spatial chromatic model (SCM) is a novel entity-oriented spatial data model proposed in recent years. Spatial tessellations formed by SCM have been applied in preliminary analyses of spatial point pattern recognition, spatial topology, and Voronoi diagram generation. Unlike previous spatial tessellation models, a key feature of SCM is that each subspace (cell) possesses a chromatic code, which serves as the foundation for various theoretical analyses and applications of SCM. The uniqueness of cell cluster codes in SCM constitutes an important property of the model. Previous studies have demonstrated the uniqueness of the chromatic codes of 2-cell and 3-cell connected disjoint clusters. Building on this analysis of connected cluster code characteristics, this paper proves the uniqueness of the chromatic codes of 4-cell connected disjoint clusters. The proof results and the new concepts introduced in this study help further refine the spatial chromatic model and provide a theoretical foundation for its subsequent applications.
\end{abstract}

\maketitle

\section{Introduction}

Spatial tessellation usually refers to partitioning a 2D or 3D space or spatial set $\boldsymbol{S}$ into many subspaces or subsets $s_{1}$, $s_{2}$, ..., $s_{3} $ without gaps and overlaps. For example, the Earth's surface is a tessellation partitioned naturally into subsets such as continents, islands, lakes, and oceans, or artificially into subsets such as countries and international regions (e.g., the high seas and Antarctica). In mathematics, spatial tessellation is related to set theory, geometry, topology, and discrete mathematics. It is also widely applied in many natural and art sciences, such as computer science (graphics and computational geometry), geographical information science (spatial analysis), chemistry (crystal structure), and graphic design (decorative figures) \cite{Zermout2019}.

Common spatial tessellation models include the following: (1) Raster model. A digital image is a type of spatial tessellation that partitions an image space into many regular pixels, with each pixel being a subspace \cite{Zavattieri1996}. In mathematics, it is a discretization of continuous space. (2) Triangular irregular network (TIN). A TIN partitions space into a number of irregular triangles (e.g., Delaunay triangles) that are connected as a network. TINs are applied in digital elevation models (DEMs), digital cartography, and 3D topographic modeling \cite{Yamada1996}, \cite{Weiler1985}. (3) Voronoi diagrams (VD). A VD is typically a point-based spatial tessellation and has been used in computer graphics (nearest-neighbor spatial queries and spatial indexing), geography (spatial interpolation), astronomy (spatial distribution of galaxies), and neuroscience (cellular structure) \cite{Okabe2000}, \cite{Gold1994}. (4) Penrose tiles. Spatial tessellation in a 2D Euclidean plane is sometimes called plane tiling, in which a subspace consists of one or a few polygons \cite{Penrose1979}. Roger Penrose, the 2020 Nobel prize winner, proposed the famous Penrose tiles, which are a pair of shapes that tile the plane aperiodically.

This study focuses on a new spatial tessellation model, called the \textbf{spatial chromatic model (SCM)}, which is an entity-oriented spatial partition analogous to Voronoi diagrams. Since the subspaces within an SCM is encoded, it is mathematically a discretization of continuous space resembling a raster model. The subspaces in an SCM are identified by their chromatic codes, which consist of sequences of integers and contain richer information about space and entities the row and column numbers used in raster models. Moreover, unlike floating-point coordinates in traditional Euclidean spaces, integer chromatic codes can be processed and computed by computers more efficiently and rapidly. This paper primarily investigates a theoretical property of the model --- the uniqueness of spatial chromatic codes.

\section{Spatial chromatic model: conception and properties}

The SCM is a novel entity-based spatial data model proposed in recent years \cite{Zhu2010}, \cite{Zhu2015}, \cite{Zhu2022}. Its fundamental principle involves mapping the relationships between spatial entities onto corresponding spatial representations, rather than placing entities within a predefined spatial container. The spatial tessellation formed in the SCM is called the spatial chromatic tessellation (SCT). The SCT contains a large number of irregular subspaces called cells. The chromatic codes of these cells can be applied to a variety of entity-space analyses, and in term of these codes, cells can be merged to form many types of spatial tessellations.

Preliminary studies conducted earlier found that SCM has considerable theoretical and practical application value, as demonstrate in the following areas: (1) Spatial topological analysis. Spatial chromatic codes are used to rapidly calculate and analyze chromatic subspaces and the spatial topological relationships between entities. (2) Spatial tessellation generation. Applying different cell-merging rules can generate various types of Voronoi diagrams and other non-Voronoi tessellations. (3) Spatial point pattern recognition. The statistical distribution characteristics of full spatial chromatic codes are used to identify the point pattern characteristics of spatial entities. (4) Spatial competition analysis. By analyzing the statistical characteristics of chromatic codes, the competitive relationships and influence distributions of entities in space can be examined. (5) Refinement of geospatial concepts. The model updates and expands the understanding of concepts such as spatial proximity and spatial heterogeneity involved in the First Law of Geography, and provides related quantitative methods.

\subsection{Generating a standard SCM}

In general, an SCM is generated by mapping entity relationships into a geometric space in which subspaces illustrate the same relationships \cite{Zhu2015}. In this study, we introduce only the standard SCM, which is based on an entity set and half-space partitions. Note that an entity is a physical concept, equivalent to an object in computer science or a point in mathematics. Suppose an point set (or entity set) $\boldsymbol{O} = \{o_{1},o_{2},\ldots ,o_{n}\}$ contains $n$ points, and the index of each point is mapped to an individual color; then the process of generating a 2D chromatic space of the standard SCM from $\boldsymbol{O}$ is as follows.

\begin{itemize}
	
\item \textbf{Step 1}. Use the perpendicular bisector $pb(o_{i},o_{j})$ (also called the dyeing line) of the line segment connecting a point-pair $\{o_{i},o_{j}\}$ to partition the 2D plane into two half-planes, $HP(o_{i})$ and $HP(o_{j})$, with respect to $o_{i}$ and $o_{j}$. Each half-plane is dyed with the color of its corresponding entity, see Fig. 1(b).

\item \textbf{Step 2}. Repeat Step 1 for all point-pairs in set $\boldsymbol{O}$, i.e., for $\binom{n}{2}$ times. The resulting $\binom{n}{2}$ dyeing lines then partition the 2D plane into many half-planes. Through this process, each half-plane is dyed with the colors of different entities, see Fig. 1(c) and 1(d).

\item \textbf{Step 3}. Overlap all half-planes generated in Step 2 to form a type of spatial tessellation containing a large number of subspaces, called \emph{cells}. Each cell has a chromatic code computed by counting how many times the cell was dyed by each entity color, and then arranging these counts in a predefined order, for example, the same order as the index set of the entity set $\boldsymbol{O}$, see Fig. 1(e) and 1(f).

\end{itemize}

\begin{figure}[h]
	\centering
	\includegraphics[width=1.0\linewidth]{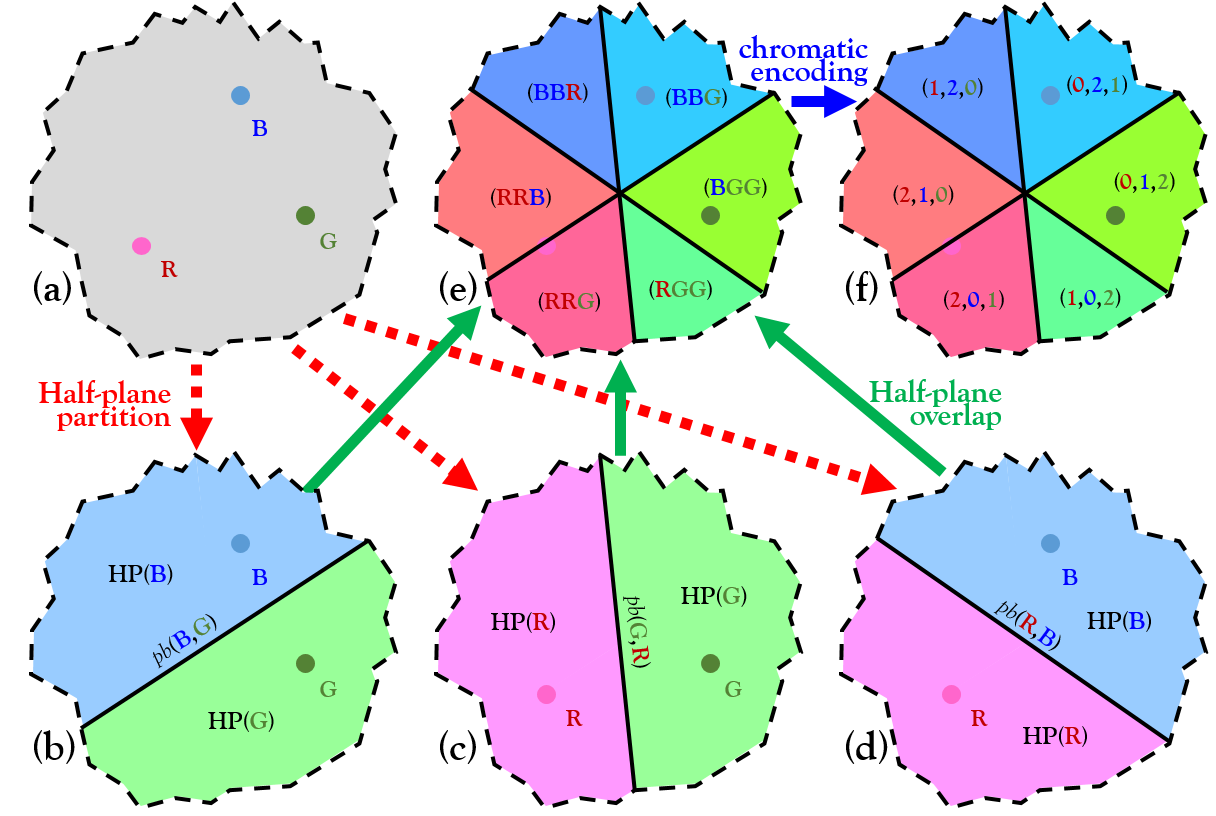}
	\caption{Generating an SCM space from three points. \textbf{(a)} Three points \{R, G, B\} in a 2D plane. \textbf{(b)}, \textbf{(c)}, and \textbf{(d)} Half-plane partitioning and dyeing with respect to the point-pairs \{B, G\}, \{R, G\}, and \{R, B\}, respectively. \textbf{(e)} The chromatic space containing six cells, generated by overlapping the half-planes in (b), (c), and (d). \textbf{(f)} Chromatic encoding: converting the cell labels in (e) to chromatic code; for example, the label (RRB) in (e) become the chromatic code (2, 1, 0) in (f).}
	\label{fig:fig1}
\end{figure}

The SCM and its chromatic space generated by the above three steps are called the \emph{standard} or \emph{ordinary} SCM, chromatic space, or spatial chromatic tessellation. The term ``standard'' or ``ordinary'' means that (1) half-plane partitioning must be applied to all point-pairs, (2) dyeing lines must be the perpendicular bisectors of the point-pairs, and (3) each half-plane must be dyed (labeled) using entity colors (indices). If one or more of these three conditions is changed, other types of spatial chromatic models, for example, the half-SCM \cite{Zhu2010}. In this work, we focus only on the standard SCM.

Fig. 2 shows the chromatic spaces generated from 4 and 6 entities. As the number of entities used to generate the space increases, the number of cells in the chromatic space increases substantially. A cell is an indivisible basic unit in an SCM and is equivalent to a pixel, a Voronoi region, an irregular triangle, or a Penrose polygon in well-known spatial tessellations, or an area in Euclidean space, or a region in geographical space. The greatest difference between a cell and these known basic spatial units is that a cell has an integer chromatic code. Spatial calculations, analyses, and applications of SCM are all based on the chromatic codes of cells. The chromatic code is similar to the integer coordinates $(i,j)$ (rows and columns) in a digital space or real-number coordinates $(x,y)$ in a Euclidean space.

\begin{figure}[h]
	\centering
	\includegraphics[width=1.0\linewidth]{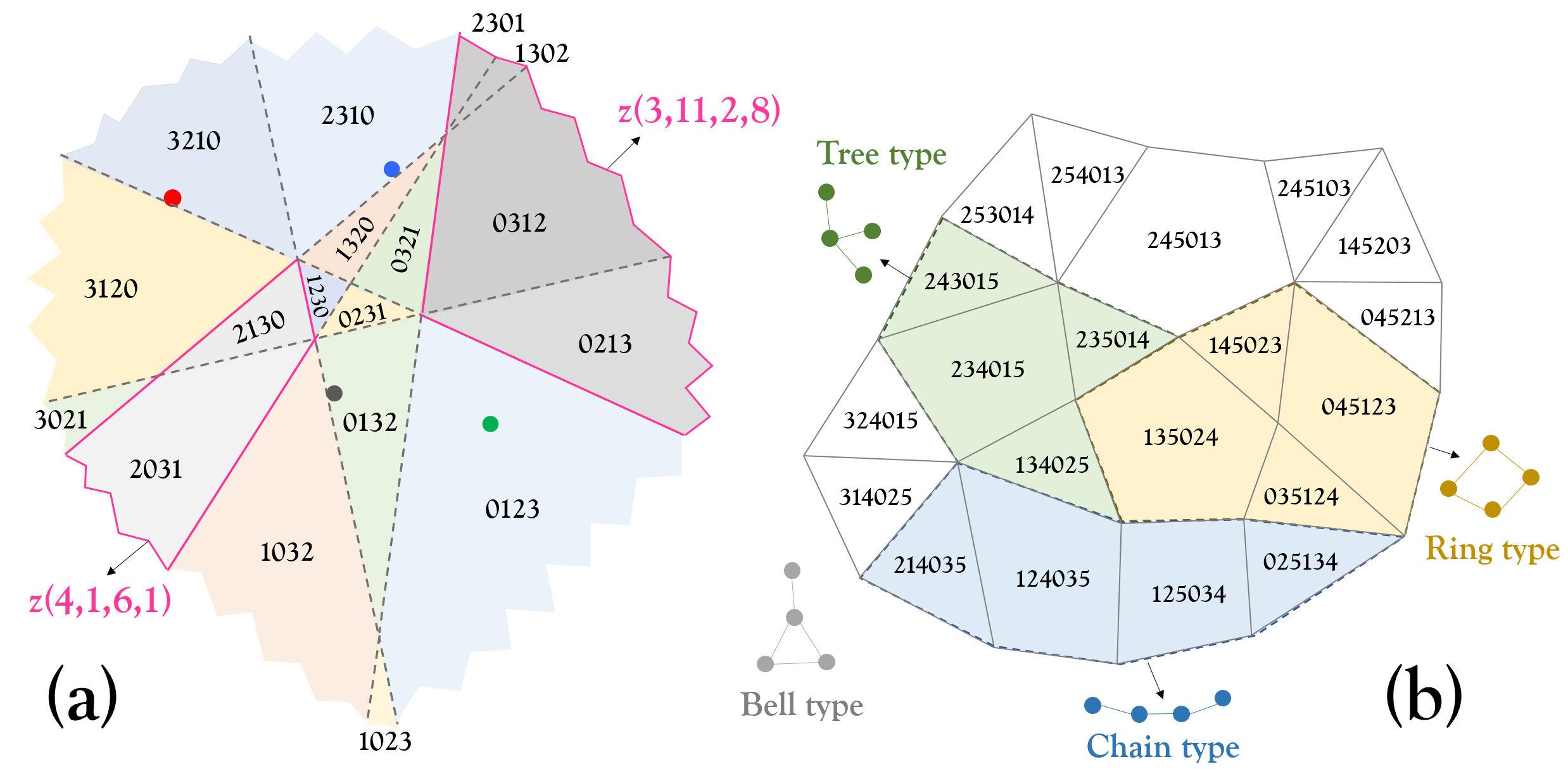}
	\caption{Examples of standard SCMs generated from \textbf{(a)} four entities and \textbf{(b)} six entities. In (a), clusters $z(4,1,6,1)$ and $z(3,11,2,8)$ are merged from two and four cells, respectively. In (b), there are three types of 4-cc clusters --- chain, tree, and ring --- but no bell type. Note that cells in (b) are the homeomorphic part of the whole standard SCM generated from six entities.}
	\label{fig:fig2}
\end{figure}

\subsection{Basic conception and properties of SCM}

We now introduce several definitions and properties of the SCM that are relevant to this study.

\begin{definition}
In an SCM generated by $n$ entities, the cellular chromatic codes are always permutations of the integer sequence $\{0,1,...,n-1\}$, which is called the base code of the cellular chromatic codes. The chromatic code of a cell $c$ is denoted by $c(a_{0},a_{1},\ldots ,a_{i},\ldots ,a_{n-1})$, where each integer $a_{i}$ is called a subcode of the chromatic code. The index $i$ of the subcode $a_{i}$ is called the location of the subcode.
\end{definition}

To highlight the subcode locations, a cell $c$ is sometimes denoted by

\begin{equation}
c\left( 
\begin{array}{cccccc}
l_{0} & l_{1} & \cdots & l_{i} & \cdots & l_{n-1} \\ 
a_{0} & a_{1} & \cdots & a_{i} & \cdots & a_{n-1}
\end{array}
\right) ,
\end{equation}

which indicates that the subcode of $c$ at location $l_{i}$ is $a_{i}$. For example, given the cell $c(1,2,4,0,3,5)$, the subcode at location 0 is 1, and the subcode at location 3 is 0. Note that in an SCM, subcode location indices typically start at 0 rather than 1.

\begin{definition}
The chromatic distance $d(c_{1},c_{2})$ between two cells $c_{1}(a_{10}$, $a_{11}$, $\ldots$, $a_{1,n-1})$ and $c_{2}(a_{21}$, $a_{22}$, $\ldots$, $a_{2,n-1})$ is defined by
\begin{equation}
d(c_{1},c_{2})=\sum\nolimits_{i=0}^{n-1}\left\vert a_{1i}-a_{2i}\right\vert
\end{equation}
\end{definition}

If cells $c_{1}$ and $c_{2}$ are treated as points in a high-dimensional Euclidean space, the chromatic distance $d(c_{1},c_{2})$ equals the $L_{1}$ (Manhattan) distance between them.

\begin{definition}
A subspace formed by merging two or more cells is called a cluster, which is equivalent to the union of those cells. A cluster formed from $m$ cells is called an $m$-cluster, and the $m$ cells are called its component cells.
\end{definition}

The chromatic code of a cluster is the sum of the chromatic codes of its component cells. Given $m$ cells with chromatic codes $c_{1}(a_{10},a_{11},\ldots ,a_{1,n-1})$, $c_{2}(a_{21}$, $a_{22}$, $\ldots$, $a_{2,n-1})$, $\ldots$, and $c_{m}(a_{m0},a_{m1},\ldots ,a_{m,n-1})$, the cluster formed by merging these $m$ cells is denoted by $z\{c_{1},c_{2},\ldots,c_{m}\}$, and its chromatic code is 
\begin{equation}
z\left( \sum\nolimits_{i=1}^{m}a_{i0},\sum\nolimits_{i=1}^{m}a_{i1},\ldots,\sum\nolimits_{i=1}^{m}a_{i,n-1}\right).
\end{equation}

For example, in Fig. 2(a), two cells $c_{1}(2,0,3,1)$ and $c_{2}(2,1,3,0)$ are merged into a cluster $z(4,1,6,1)$, and four cells $c_{1}(2,3,0,1)$, $c_{2}(1,3,0,2)$, $c_{3}(0,3,1,2)$, and $c_{4}(0,2,1,3)$ are merged into a cluster $z(3,11,2,8)$. In Fig. 2(b), four cells $c_{1}(2,1,4,0,3,5)$, $c_{2}(1,2,4,0,3,5)$, $c_{3}(1,2,5,0,3,4)$, and $c_{4}(0,2,5,1,3,4)$ are merged into a chain type 4-cell cluster $z(4,7,18,1,12,18)$.

The chromatic code of a cell is always a permutation of the integer sequence $\{0,1,\ldots ,n-1\}$, and hence, if we swap two subcodes in the chromatic code of a cell, we obtain the chromatic code of another cell. 

\begin{definition}
For a cell $c(a_{0},a_{1},\ldots,a_{i},\ldots,a_{j},\ldots,a_{n-1})$, swapping two subcodes $a_{i}$ and $a_{j}$, regardless of their locations, is called a value transition between $a_{i}$ and $a_{j}$ and is denoted by $TE(a_{i},a_{j})$. Swapping two subcodes at locations $l_{i}$ and $l_{j}$, regardless of their values, is called a location transition in a chromatic code and is denoted by $TL(l_{i},l_{j})$. 	
\end{definition}

Applying $TE(a_{i},a_{j})$ or $TL(l_{i},l_{j})$ to a cell $c$ yields a new cell $c^\prime$, denoted by $c\times TE(a_{i},a_{j})=c^\prime$ or $c\times TL(l_{i},l_{j})=c^\prime$. For example, cell $c(1,0,3,4,5,2)\times TE(2,4) = c^\prime(1,0,3,2,5,4)$, i.e., swapping the subcodes 2 and 4 of $c$ yields the cell $c^\prime$; $c(1,0,3,4,5,2)\times TL(2,4) = c^\prime(1,0,5,4,3,2)$, i.e., swapping the subcodes at locations 2 and 4 (where the subcode values are 3 and 5, respectively) of $c$ yields the cell $c^\prime$. This example also indicates that usually  $c\times TE(a,b) \ne c\times TL(a,b)$.

\begin{definition}
Two cells $c_{1}$ and $c_{2}$ in a 2D plane are called neighboring or adjacent if they share a common segment of a dyeing line; they are then mutual neighbors, or connected, and are denoted by $c_{1}\leftrightarrow c_{2}$.	
\end{definition}

In a 2D Euclidean space, two neighboring cells correspond to two convex polygons sharing a common edge. Given a cell $c_{1}$, all its neighboring cells are collectively analogous to the neighborhood of a point in point-set topology. 

That cell $c_{2}$ is connected to cell $c_{1}$ is sometimes denoted by $c_{1}\rightarrow c_{2}$ and is described as $c_{1}$ adjacently transitioning to $c_{2}$, or as $c_{1}$ taking one step to $c_{2}$. If, through a sequence of adjacent transitions (after many steps), the chromatic code of $c_{1}$ becomes that of $c_{2}$, then $c_{1}$ and $c_{2}$ are called \emph{path-connected}, and all cells involved in these transitions form a \emph{path} between $c_{1}$ and $c_{2}$.

\begin{definition}
In an $m$-cell cluster, if any two component cells are path-connected and all cells along the path are component cells of the cluster, then the $m$-cell cluster is called a connected cluster and is denoted by a $m$-cc cluster.
\end{definition}

\begin{definition}
Two cells $c_{1}$($a_{10}$, $a_{11}$, \ldots, $a_{1i}$, \ldots, $a_{1,n-1}$) and $c_{2}$($a_{20}$, $a_{21}$, \ldots, $a_{2i}$, \ldots, $a_{2,n-1}$) are said to be equal, denoted by $c_{1}=c_{2}$, if their corresponding subcodes are all equal, that is, $a_{1i}=a_{2i}$, for each $i$.
\end{definition}

Therefore, if at any one location their subcodes are not equal, then the two cells are not equal, denoted by $c_{1}\neq c_{2}$. Equality and inequality between two clusters are defined similarly.

\begin{property}
The chromatic distance between two connected cells $c_{1}$ and $c_{2}$ is 2, and the chromatic distances between two equal cells or clusters are both 0, that is,
\begin{equation}
	d(c_{1},c_{2})=2, \text{ if } c_{1} \leftrightarrow c_{2},
\end{equation}
and
\begin{equation}
	d(c_{1},c_{2})=0 \text{ or } d(z_{1},z_{2})=0, \text{ if } c_{1}=c_{2} \text{ or } z_{1}=z_{2}.
\end{equation}
\end{property}

\begin{property}
If the same location transitions $TL(l_{i},l_{j})$ are applied respectively to all component cells of a cluster, then the resulting cells remain their connectivity after these transitions.
\end{property}
This property indicates that if any two component cells $c_{1}\leftrightarrow c_{2}$, then after $c_{1}\times TL(l_{i},l_{j})=c_{1}^\prime$, $c_{2}\times TL(l_{i},l_{j})=c_{2}^\prime$, the two cells $c_{1}^\prime\leftrightarrow c_{2}^\prime$; or if $c_{1}\nleftrightarrow c_{2}$, then $c_{1}^\prime\nleftrightarrow c_{2}^\prime$.

\begin{definition}
If two clusters $z_{1}$ and $z_{2}$ share at least one component cell, that is, at least one component cell of $z_{1}$ equals a component cell of $z_{2}$, then they are called joint, denoted by $z_{1}\cap z_{2}\neq \varnothing $. If $z_{1}\cap z_{2}=\varnothing $, they are disjoint. If two or more $m$-cc clusters are pairwise disjoint, they are called disjoint $m$-cc clusters, denoted as $m$-ccd clusters.
\end{definition}

\begin{property}
In a standard SCM, the chromatic codes of cells are unique; that is, no two cells are equal. The chromatic codes of connected disjoint 2-cell and 3-cell clusters are also unique, that is, no two 2-ccd or 3-ccd clusters are equal.
\end{property}

The above Properties 1, 2, and 3 were proved in \cite{Zhu2010}, \cite{Zhu2015}.

\section{Code uniqueness of 4-ccd clusters}

The uniqueness of the chromatic codes of 2-ccd and 3-ccd clusters has been proved \cite{Zhu2015}. For clusters formed by more than three cells ($m>3$), it remains a conjecture that the chromatic codes of $m$-ccd clusters are also unique. In this section, we prove the uniqueness of the chromatic codes of 4-ccd clusters, namely, as stated in the following theorem.

\begin{theorem}
	If two clusters $z_{1}$ and $z_{2}$ are 4-ccd clusters, then $z_{1} \neq z_{2}$. 
\end{theorem}

The main idea of our proof is to first analyze the features of all possible 4-cc clusters and then prove that if two 4-cc clusters are disjoint, their chromatic codes are not equal; if the chromatic codes of two 4-cc clusters are equal, then the two clusters must be joint.

\subsection{Connection types and subcode transitions of 4-cc clusters}

According to the definition of a 4-cc cluster, it contains four connected cells $c_{1}$, $c_{2}$, $c_{3}$, and $c_{4}$. The chromatic codes of these four cells must undergo three swaps of two subcodes, and their connections are illustrated in Fig. 2(b). At most six subcodes, $a_{0}$, $a_{1}$, $a_{2}$, $a_{3}$, $a_{4}$, and $a_{5}$, are involved in the three swaps, which are as follows: (1) starting from cell $c_{1}$, swap its subcode $a_{0}$ at location $l_{0}$ with the subcode $a_{1}$ at location $l_{1}$, denoted by $c_{1}\rightarrow c_{2}$ (or $c_{1}$ walks to $c_{2}$); (2) starting from cell $c_{2}$, swap its subcode $a_{2}$ at location $l_{2}$ with the subcode $a_{3}$ at location $l_{3}$, denoted by $c_{2}\rightarrow c_{3}$ (or $c_{2}$ walks to $c_{3}$); (3) starting from cell $c_{3}$, swap its subcode $a_{4}$ at location $l_{4}$ with the subcode $a_{5}$ at location $l_{5}$, denoted by $c_{3}\rightarrow c_{4}$ (or $c_{3}$ walks to $c_{4}$). This sequence of walks from $c_{1}$ to $c_{4}$ is written as
\begin{gather}\label{eq:walks}
c_{1}\left( 
\begin{array}{ccccccccccccc}
\cdots & l_{0} & \cdots & l_{1} & \cdots & l_{2} & \cdots & l_{3} & \cdots & 
l_{4} & \cdots & l_{5} & \cdots \\ 
\cdots & a_{0} & \cdots & a_{1} & \cdots & a_{2} & \cdots & a_{3} & \cdots & 
a_{4} & \cdots & a_{5} & \cdots%
\end{array}%
\right) \\
\times TE(a_{0},a_{1})\rightarrow  \notag \\
c_{2}\left( 
\begin{array}{ccccccccccccc}
\cdots & l_{0} & \cdots & l_{1} & \cdots & l_{2} & \cdots & l_{3} & \cdots & 
l_{4} & \cdots & l_{5} & \cdots \\ 
\cdots & a_{1} & \cdots & a_{0} & \cdots & a_{2} & \cdots & a_{3} & \cdots & 
a_{4} & \cdots & a_{5} & \cdots%
\end{array}%
\right)  \notag \\
\times TE(a_{2},a_{3})\rightarrow  \notag \\
c_{3}\left( 
\begin{array}{ccccccccccccc}
\cdots & l_{0} & \cdots & l_{1} & \cdots & l_{2} & \cdots & l_{3} & \cdots & 
l_{4} & \cdots & l_{5} & \cdots \\ 
\cdots & a_{1} & \cdots & a_{0} & \cdots & a_{3} & \cdots & a_{2} & \cdots & 
a_{4} & \cdots & a_{5} & \cdots%
\end{array}%
\right)  \notag \\
\times TE(a_{4},a_{5})\rightarrow  \notag \\
c_{4}\left( 
\begin{array}{ccccccccccccc}
\cdots & l_{0} & \cdots & l_{1} & \cdots & l_{2} & \cdots & l_{3} & \cdots & 
l_{4} & \cdots & l_{5} & \cdots \\ 
\cdots & a_{1} & \cdots & a_{0} & \cdots & a_{3} & \cdots & a_{2} & \cdots & 
a_{5} & \cdots & a_{4} & \cdots%
\end{array}%
\right),  \notag
\end{gather}
where the ellipses indicate the locations at which the subcodes are not swapped; these are called \emph{static subcodes}, and the remaining subcodes that are involved in swaps are called \emph{kinetic subcodes} of a cell or a cluster. For example, the chromatic code of cell $c_{4}$ can be written as
\begin{equation}
c_{4}\left( 
\begin{array}{cccccc}
l_{0} & l_{1} & l_{2} & l_{3} & l_{4} & l_{5} \\ 
a_{1} & a_{0} & a_{3} & a_{2} & a_{5} & a_{4}
\end{array}
\right) \cup \left( 
\begin{array}{c}
L_{sta} \\ 
A_{sta}
\end{array}
\right),
\end{equation}
where the first and second parts in brackets denote the kinetic and static codes, respectively. In some cases, the static codes may be omitted.

Accordingly, the cluster $z$ merged from the four cells in Eq.(\ref{eq:walks}) has the chromatic code partitioned by the kinetic and static codes as the following,
\begin{equation}
z\left( 
\begin{array}{cccccc}
l_{0} & l_{1} & l_{2} & l_{3} & l_{4} & l_{5} \\ 
a_{0}+3a_{1} & a_{1}+3a_{0} & 2a_{2}+2a_{3} & 2a_{3}+2a_{2} & 3a_{4}+a_{5} & 
3a_{5}+a_{4}
\end{array}
\right) \cup \left( 
\begin{array}{c}
L_{sta} \\ 
A_{sta}
\end{array}
\right).
\end{equation}

According to \textbf{Property 1}, value transitions between two connected cells only swap two integers $a$ and $a+1$; therefore, the kinetic codes of a 4-cc cluster, namely $a_{0}$ to $a_{5}$ in Eq. (\ref{eq:walks}), are limited to the following cases: (1) $\left\vert a_{1}-a_{0}\right\vert =1$, $\left\vert a_{3}-a_{2}\right\vert =1$, and $\left\vert a_{5}-a_{4}\right\vert =1$; (2) for a subcode $a$, it can swap with either $a+1$ or $a-1$; (3) a value transition can be applied twice along a path, provided that the two applications are not consecutive. For example, by $TE(0,1)$, the cell $c_{1} \rightarrow c_{2}$. If apply immediately $TE(0,1)$ again, then $c_{2} \leftarrow c_{1}$ without connecting to a new cell. However, if we apply $TE(2,3)$ to $c_{2}$, making it connected to $c_{3}$, and then apply $TE(0,1)$ to $c_{3}$, the latter will connect to a new cell $c_{4}$ instead of walking back to $c_{1}$.

Based on the above three limitations, the value transitions in a 4-cc cluster must be as follows:

\begin{enumerate}
	\item Three subcodes swapped. They must be $a$, $a+1$, and $a+2$ (three consecutive integers), denoted by ($a\oplus3\leftrightarrow3$), meaning that the swaps occurred three times ($\leftrightarrow3$) among the three consecutive integers starting from $a$ ($a\oplus3$).
	\item Four subcodes swapped. They must be $a$, $a+1$, $b$, and $b+1$ ($a+1<b$). If $b=a+2$, the four subcodes are $a$, $a+1$, $a+2$, and $a+3$. This case is denoted by ($a\oplus2\leftrightarrow1)(b\oplus2\leftrightarrow2$), ($a\oplus2\leftrightarrow2)(b\oplus2\leftrightarrow1$) or ($a\otimes4\leftrightarrow3$).
	\item Five subcodes swapped. They must be $a$, $a+1$, $a+2$, $b$, and $b+1$ ($b>a+3$) or $a$, $a+1$, $b$, $b+1$, $b+2$ ($b>a+2$). If $b=a+3$ or $b=a+2$, then the five subcodes are $a$, $a+1$, $a+2$, $a+3$, $a+4$, i.e., five consecutive integers. This case is denoted by ($a\oplus3\leftrightarrow2)(b\oplus2\leftrightarrow1$), ($a\oplus2\leftrightarrow1)(b\oplus3\leftrightarrow2$) or ($a\oplus5\leftrightarrow3$).
	\item Six subcodes swapped. They must be $a$, $a+1$, $b$, $b+1$, $e$, and $e+1 $ ($b>a+2$ and $e>b+2$). If $b=a+2$, then the six subcodes are $a$, $a+1$, $a+2$, $a+3$, $e$, and $e+1$. If $e=b+2=a+4$, then the six subcodes are $a$, $a+1$, $a+2$, $a+3$, $a+4$, and $a+5$, i.e., six consecutive integers. If $b>a+2$ and $e=b+2$, then the six subcodes are $a$, $a+1$, $b$, $b+1$, $b+2$, and $b+3$. This case is denoted by ($a\oplus2\leftrightarrow1)(b\oplus2\leftrightarrow1)(e\oplus2\leftrightarrow1$), ($a\oplus4\leftrightarrow2)(b\oplus2\leftrightarrow1$), ($a\oplus2\leftrightarrow1)(b\oplus4\leftrightarrow2$), or ($a\oplus6\leftrightarrow3$). There is no such case that $a$, $a+1$, $a+2$, $b$, $b+1$, $b+2$, because the swaps occur only three times.  	
\end{enumerate}

A consecutive integer interval over which subcodes are involved in swaps is called an \emph{transition interval}. Subcode swaps occur only within each transition interval, and subcodes in different transition intervals do not swap with each other. In a 4-cc cluster, there are only ten types of subcode transition intervals, as listed in the above four cases, such as ($a\oplus3\leftrightarrow3$), ($a\oplus2\leftrightarrow1)(b\oplus3\leftrightarrow2$), and ($a\oplus2\leftrightarrow1)(b\oplus2\leftrightarrow1)(e\oplus2\leftrightarrow1$).

If each cell is shrunk to a node and neighboring cells are connected by an edge, then a 4-cc cluster is converted into a connected graph, as shown in Fig. 2(b). Such a graph can be of tree type, chain type, and ring type (a chain type in which the head and tail are connected). Note that because no three cells in an SCM form a ring, the bell type 4-cc clusters (see an example in Fig. 2(b)) do not occur.

The above four cases and three connection types demonstrates that, in a 4-cc cluster, the subcode transitions and neighbor relations among the four component cells are diverse and complicated. We first analyze the case in which the transition interval is an integer sequence containing six consecutive integers $(a$, $a+1$, $a+2$, $a+3$, $a+4$, $a+5)$; the cell bears the six integers is called the \emph{original cell} or \emph{seed cell}. Since subcode value transitions occur only within the six subcodes, we call them the \emph{transition codes} or \emph{t-code} of the cell, in which $a$ is the minimum and is called the \emph{base} of transition codes. If the base is omitted, the t-code $(a$, $a+1$, $a+2$, $a+3$, $a+4$, $a+5)$ is usually shorten to $(0,1,2,3,4,5)$.

Since the base of t-codes is $a$, the original cell must undergo a value transition $TE(a, a+1)$, called the \emph{minimal transition}. That is, applying $TE(a, a+1)$ to the original cell $(a$, $a+1$, $a+2$, $a+3$, $a+4$, $a+5)$ connects it to the neighboring cell $(a+1$, $a$, $a+2$, $a+3$, $a+4$, $a+5)$, i.e., $(0,1,2,3,4,5)\rightarrow (1,0,2,3,4,5)$ (see the example in Fig. 3, where $c_{0}\rightarrow c_{1}$). Then, applying $TE(2,3)$ gives $c_{1} \rightarrow c_{2}(1,0,3,2,4,5)$, and applying $TE(4,5)$ gives $c_{2} \rightarrow c_{3}(1,0,3,2,5,4)$. The four cells along the connection path $c_{0} \rightarrow c_{1} \rightarrow c_{2} \rightarrow c_{3}$ merge into a 4-cc cluster $z$ with t-code $(4a+3$, $4a+1$, $4a+10$, $4a+10$, $4a+17$, $4a+19)\cup (A_{sta})$. If we omit $4a$ in the kinetic code as well as the static code $A_{sta}$, the t-code of $z$ becomes $ (3,1,10,10,17,19)$. The similar t-codes of other 4-cc clusters are shown in Fig. 3.

As illustrated by cluster $z$ above, starting from an original cell $c_{0}$ and proceeding through successive transitions (path connections), the resulting cluster $z$ is called a \emph{generated cluster} of $c_{0}$. The set of all generated from $c_{0}$ is called the \emph{generated cluster family} of $c_{0}$ (see all generated 4-cc clusters in Fig. 3).

\begin{figure}[h]
	\centering
	\includegraphics[width=1.0\linewidth]{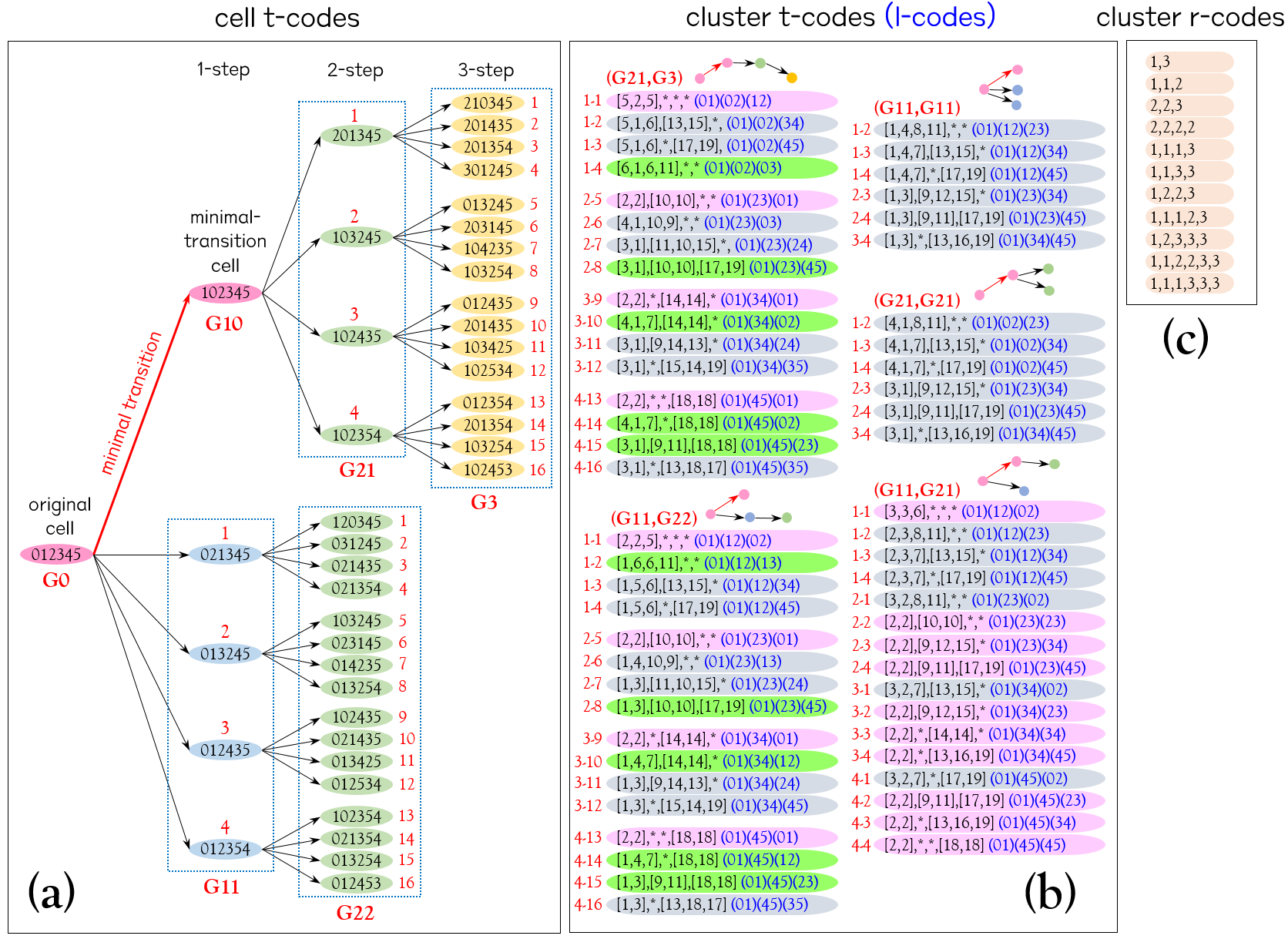}
	\caption{The 4-cc cluster family generated from the original cell $c_{0}(0,1,2,3,4,5)$. \textbf{(a)} The neighboring cells and their groups of $c_{0}$. \textbf{(b)} The t-codes and l-codes of the generated 4-cc clusters. \textbf{(c)} The residue codes (r-codes) of the t-codes.}
	\label{fig:fig3}
\end{figure}

The t-code $(0,1,2,3,4,5)$ contains six integers, and therefore only five types of value transitions can occur between neighboring cells, namely, $TE(0,1)$, $TE(1,2)$, $TE(2,3)$, $TE(3,4)$, and $TE(4,5)$. Each cell can thus connect to at most five neighboring cells. Fig. 3 shows the original cell $c_{0}(0,1,2,3,4,5)$, the cell $c_{1}(1,0,2,3,4,5)$, which is connected to $c_{0}$ by the minimal transition, and all neighboring cells that are path-connected to either $c_{0} $ or $c_{1}$.

The value transitions occurred three times between the component cells of a 4-cc cluster, and the minimal transition between $c_{0}$ and $c_{1}$ must be included. Therefore, the three value transitions correspond to the three steps shown in Fig. 3. These steps are analogous to the growth of a tree from its root (the original or seed cell). We categorize these cells into several \emph{neighboring groups} as follows:
\begin{itemize}
	\item \textbf{G0}: the original cell $c_{0}(0,1,2,3,4,5)$ only.
	\item \textbf{G10}: the minimal-transition cell $c_{1}(1,0,2,3,4,5)$ only.
	\item \textbf{G11}: the cells that are connected to $c_{0}$, excluding $c_{1}$.
	\item \textbf{G21}: the cells that are connected to the cells in \textbf{G10}.
	\item \textbf{G22}: the cells that are connected to the cells in \textbf{G11}.
	\item \textbf{G3}: the cells that are connected to the cells in \textbf{G21}.
\end{itemize}

Starting from groups \textbf{G0} and \textbf{G10}, the cell connections that occur in a 4-cc cluster can take the following forms:
\begin{itemize}
	\item (\textbf{G21}, \textbf{G3}): connecting a cell in \textbf{G21} and then a cell in \textbf{G3}.
	\item (\textbf{G11}, \textbf{G21}): connecting a cell in \textbf{G11} and then a cell in \textbf{G21}.
	\item (\textbf{G11}, \textbf{G22}): connecting a cell in \textbf{G11} and then a cell in \textbf{G22}.
	\item (\textbf{G11}, \textbf{G11}): connecting two cells in \textbf{G11}.
	\item (\textbf{G21}, \textbf{G21}): connecting two cells in \textbf{G21}.
\end{itemize}

Note that a few cells appear in more than one neighboring group. For example, $(0,1,2,3,5,4)$ is in both \textbf{G1} and \textbf{G3}, and $(1,0,2,3,5,4)$ is in both \textbf{G21} and \textbf{G22} (see Fig. 3).

Starting from the original cell $(0,1,2,3,4,5)$ and applying the minimal transition $TE(0,1)$ together with all possible path connects, we list the generated 4-cc cluster family and their t-codes in Fig. 3.

\subsubsection{Features of chromatic codes of 4-ccd clusters}

The t-code of a 4-cc cluster is computed by summing the t-codes of all its component cells. For example, summing the t-codes of the following four cells $c_{1}$, $c_{2}$, $c_{3}$, and $c_{4}$, yields a cluster $z$ with t-code $(5,2,5,12,16,20)$.
\begin{eqnarray} \label{eq:tcode_example}
&&c_{0}(0,1,2,3,4,5) \\
&&c_{1}(1,0,2,3,4,5)  \notag \\
&&c_{2}(2,0,1,3,4,5)  \notag \\
&&\underline{c_{3}(2,1,0,3,4,5)}  \notag \\
&&z(5,2,5,12,16,20)  \notag
\end{eqnarray}

Since a value transition $TE(a,b)$ always corresponds to a location transition $TL(s,t)$, then the t-code generated by three value transitions in a 4-cc cluster also corresponds to three location transitions. For example, the above t-code (5,2,5,12,16,20) is obtained through three value transitions $TE(0,1)$, $TE(1,2)$, and $TE(0,1)$, which correspond to location transitions $TL(0,1)$, $TL(0,2)$, and $TL(1,2)$, respectively. We rewrite the three location transitions as  $(0,1)(0,2)(1,2)$, and call this notation the \emph{location code} or \emph{l-code} of the cluster. The l-codes for the generated 4-cc cluster family are shown in Fig. 3. Each component of the l-code $(0,1)(0,2)(1,2)$, such as $(0,1)$,  is called a subcode of the l-code, and the order of these subcodes within the l-code is not fixed. 

We can see that, among the three value transitions in Eq.(\ref{eq:tcode_example}), the three subcodes $(3,4,5)$ remain unchanged, indicating they are static subcodes. Therefore, we rewrite the t-code of $z$ as $(5,2,5,\ast ,\ast ,\ast )$, where each $\ast $ denotes a static subcode position, and further abbreviate it to $(5,2,5)$.

Although we are analyzing the case of six consecutive integers $(a$, $a+1$, 
$a+2$, $a+3$, $a+4$, $a+5)$, other cases involving shorter t-codes and more transition intervals are also covered by the same t-codes in Fig. 3:

\begin{itemize}
	\item Cases in which 3, 4, 5, or 6 subcodes are swapped. In Fig. 3, static subcodes are labeled by $\ast $; therefore, the number of integers remaining in a cluster t-code equals the number of subcodes swapped in the three value transitions.
	\item Cases of connected graph types. In Fig. 3, the connection groups (\textbf{G1}, \textbf{G1}), (\textbf{G21}, \textbf{G21}) correspond to the tree type in Fig. 2(b). The connections of (\textbf{G21}, \textbf{G3}), (\textbf{G1}, \textbf{G22}), and (\textbf{G1}, \textbf{G21}) correspond to the chain and ring types in Fig. 2(b). In some chain-type clusters, the last cell after three value transitions is identical to $c_{0}$, for example, the cell $(0,1,2,3,4,5)$ occurs again in group \textbf{G3}, and therefore a ring-type cluster is formed.
	\item Cases of different transition intervals and bases. The kinetic subcode values and static subcode locations of cluster t-codes indicate the corresponding transition intervals. For example, $(5,2,5,\ast ,\ast ,\ast )$ and $(3,3,6,\ast ,\ast,\ast )$ correspond to the $(a\oplus3\leftrightarrow3)$ interval. The t-codes $(4,1,10,9,\ast,\ast)$ and $(2,2,\ast ,14,14,\ast )$ correspond to the $(a\oplus4\leftrightarrow3)$ and $(a\oplus2\leftrightarrow2)(b\oplus2\leftrightarrow1)$ intervals, respectively. It can be observed that subcode values vary significantly across different intervals, whereas within the same interval they exhibit minimal variation and may even be identical in some cases. The kinetic subcodes on either side of one or more static subcodes belong to different transition intervals. The transition intervals contained in each 4-cc cluster t-code are indicated by the square brackets in Fig. 3.  
\end{itemize}

More precisely, a t-code of a cell can be written as 
\begin{equation}\label{eq:PreciseTcode}
(a, a+1, b+a+2, b+a+3, e+a+4, e+a+5),
\end{equation}
where $a$, $b$, and $e$ are nonnegative integers and $e\ge b$. If $b=e=0$, the t-code reduces to the case of six consecutive integers. 	

Given the t-code $(t_{1},t_{2},\ldots,t_{n})$ of an $m$-cc cluster, we can compute its congruence code as follows:
\begin{equation}\label{eq:CongruenceCode}
	(t_{1} (\text{mod }m),t_{2} (\text{mod }m),\ldots,t_{n} (\text{mod }m)).
\end{equation}
For example, according to Eq. (\ref{eq:CongruenceCode}), the congruence code of the 4-cc cluster t-code $(5,2,5,\ast,\ast,\ast)$ is $(1,2,1,0,0,0)$. Since the subcodes at the $\ast$ locations are static subcodes, their chromatic codes must be of the form $4\ast$; hence their residues modulo 4 must be 0. Removing the zeros from $(1,2,1,0,0,0)$ and sorting the remaining integers in ascending order yields a new code $(1,1,2)$, which is called the \emph{residue code} or \emph{r-code} of the 4-cc cluster. The r-codes of all 4-cc clusters are shown in Fig. 3(c). Note that the r-codes in Fig. 3(c) are computed from the original cell with six consecutive integers. If the t-codes of the original cell takes the form given in Eq. (\ref{eq:PreciseTcode}), the generated clusters have the same r-codes as those in Fig. 3(c). 

\begin{property}
If the r-codes of two 4-cc clusters are not equal, then their chromatic codes are also not equal.
\end{property}

\begin{proof}
	First, it is easy to prove the following relationship
	\begin{equation} \label{eq:amodb}
		a\not\equiv b \text{ (mod 4)}\Rightarrow a\neq b,
	\end{equation}
	that is, if two integers have different residues modulo 4, then the two integer are not equal. Therefore, if the chromatic codes of two 4-cc clusters are equal, their r-codes must be equal. Conversely, if their r-codes are not equal, then by Eq. (\ref{eq:amodb}), their subcodes cannot all be equal, and hence their chromatic codes are not equal. 
\end{proof}

For example, the r-codes of clusters $z_{1}$ and $z_{2}$ are $(1,3)$ and $(1,1,2)$, respectively, so $z_{1}$ and $z_{2}$ are not equal. The number 3 in the r-code base of $z_{1}(1,3)$ indicates a location where the subcode of $z_{1}$ modulo 4 is 3. In contrast, no such subcode exists in $z_{2}$, because the r-code of $z_{2}$ is $(1,1,2)$, indicating that there are two subcodes congruent to 1 modulo 4, one subcode congruent to 2 module 4, and the remaining subcodes congruent to 0 modulo 4, but no subcode congruent to 3 modulo 4.

If the r-codes of two clusters are not equal, then the two clusters are not equal. However, if their r-codes are equal, the two clusters may or may not be equal. Therefore, it remains only to prove that any two 4-ccd clusters with the same r-codes are also not equal. This is equivalent to proving that if two 4-cc clusters have both the same chromatic code and r-code, they cannot be 4-ccd clusters, that is, they must be joint.

Note that the t-codes in Fig. 3 are only those generated from the original cell, and the number of their types are $k_{1}=60$ (including a few repeated cases). In a chromatic space, any cell can serve as the original cell, and the t-code of the new original cell is a permutation of the t-code $(0,1,2,3,4,5)$, and therefore the number of such types is $k_{2}=6!=720$. Moreover, the t-code base $a$ ranges from a minimum 0 to a maximum $n-6$, giving $k_{3}=n-5$ possible types. Given a chromatic space generated from $n>6$ entities, the number of types of chromatic codes of 4-cc clusters is $k=k_{1}\times k_{2}\times k_{3}=60\times 720\times (n-5)$. Because of the repeated cases in Fig. 3 and because some cells are hidden in the 2D plane \cite{Zhu2015}, the actual number of 4-cc clusters that can emerge in 2D space will be much lower than $k$.

\subsection{Proof of code uniqueness of 4-ccd clusters}

We now prove that, among the above $k$ types of 4-cc clusters, any two 4-ccd clusters $z_{a}$ and $z_{s}$ who are generated respectively from two original cells with the following t-codes are not equal.
\begin{eqnarray} \label{eq:ca0cs0}
&&c_{a0}\left( 
\begin{array}{cccccc}
l_{a0} & l_{a1} & l_{a2} & l_{a3} & l_{a4} & l_{a5} \\ 
a & a+1 & a+2 & a+3 & a+4 & a+5
\end{array}
\right) \cup \left( 
\begin{array}{c}
l_{ax} \\ 
x
\end{array}
\right) \\
&&c_{s0}\left( 
\begin{array}{cccccc}
l_{s0} & l_{s1} & l_{s2} & l_{s3} & l_{s4} & l_{s5} \\ 
s & s+1 & s+2 & s+3 & s+4 & s+5
\end{array}
\right) \cup \left( 
\begin{array}{c}
l_{sy} \\ 
y
\end{array}
\right)  \notag
\end{eqnarray}

The proof considers three cases for the relationship between $c_{a0}$ and $c_{s0}$:
\begin{enumerate}
	\item \textbf{Case (1)}: $a=s$ and $l_{ai}=l_{si}$.
	\item \textbf{Case (2)}: $a\neq s$.
	\item \textbf{Case (3)}: $a=s$ but $l_{ai}\neq l_{si}$.	 
\end{enumerate}
 
\textbf{Case (1)}. This case corresponds to $c_{a0}=c_{s0}$, meaning that the clusters $z_{a}$ and $z_{s}$ are generated from the same original cell $c_{a0}$. For example, all clusters in Fig. 3 are generated from the cell $(0,1,2,3,4,5)$. We can see that some generated clusters are not equal because their r-codes differ, whereas other clusters share the same t-codes. For example, there are three clusters with t-code $(2,2,10,10)$ in Fig. 3. However, these clusters are not disjoint, because any two of them are joint with both the original cell $c_{a0}$ and its minimal-transition cell $c_{a1}$. Therefore, clusters in this case are excluded from our proof, since they are not 4-ccd clusters.

\textbf{Case (2)}. This case corresponds to the two original cells having different bases. Given $a\neq s$, we may assume without loss of generality that $s>a$. The locations $l_{ai}$ and $l_{si}$ may or may not be equal. According to the definition of t-code base, $s$ is the minimal value of the subcode of t-code, since $a<s$, therefore among the generated clusters of $c_{s0}$, the subcode $a$ has not involved in value transitions (it is a static code). Suppose the location of the subcode $a$ in $c_{s0}$ is $l_{sa}$, then we consider the following cases:  

\textbf{Case (2-1)}. $l_{sa} = L_{a}$, that is, $l_{sa}$ equals one of the locations $l_{a0},\ldots, l_{a5}$ in Eq. (\ref{eq:ca0cs0}). At location $l_{sa}$, the chromatic subcode of the generated cluster $z_{s}$ is $4a$. However, at the locations $L_{a}$, the minimum subcode value of the generated cluster $z_{a}$ is $4a+1$, which occurs at location $l_{a0}$, because subcode $a$ has been swapped with $a+1$ at least once.  

\textbf{Case (2-2)}. $l_{sa} \ne L_{a}$, that is, $l_{sa}$ equals one of the locations $l_{ax}$ in Eq. (\ref{eq:ca0cs0}), at which the subcode of $c_{a0}$ is a static subcode. At any such location $l_{ax}$, the subcode of $c_{a0}$ is either $\ge a+6$ or $\le a-1$, and consequently, the subcodes of the generated cluster $z_{a}$ at these locations are either $\ge 4a+24$ or $\le 4a-4$.  

 According to above two cases, $z_{s}$ contains a subcode equal to $4a$, whereas no such subcode exists in $z_{a}$, because all subcodes of $z_{a}$ are either greater than $4a$ or less than $4a$. Therefore, $z_{a}\neq z_{s}$.

\textbf{Case (3)}. In this case, $a=s$ but $l_{ai}\ne l_{si}$, meaning that in Eq. (\ref{eq:ca0cs0}), there is at least one index $i$ such that $l_{ai} \ne l_{si}$. According to \textbf{Property 2}, we can permute $c_{a0}$ and $c_{s0}$ respectively so that both have the same order as the original cell in Fig. 3, and these permutations (location transitions) do not change the connectivity among the component cells of the generated 4-cc cluster. Therefore, we consider only the t-codes of $c_{a0}$ and $z_{a}$ shown in Fig. 3, and the t-codes of $c_{s0}$ and $z_{s}$ obtained by permuting them, as the following cases. For example, if the original cell $c_{a0}(0,1,2,3,4,5)\times TL(1,3)=c_{s0}(0,3,2,1,4,5)$, then $z_{a}(1,3,9,11,17,19)$ takes the same location transition and yields the corresponding cluster $z_{s}(1,11,9,3,17,19)$.  

\textbf{Case (3-1)}. The subcodes of the t-codes of $z_{a}$ and $z_{s}$ are all distinct; that is, no two subcodes of a t-code are equal. For example, the six subcodes of the t-code $(1,3,9,11,17,19)$ are all distinct, whereas the t-code $(5,2,5,\ast,\ast,\ast)$ does not have the property, because the subcode 5 appears twice in it.

Suppose the t-code of $z_{a}$ is ($A_{0}$, $A_{1}$, $A_{2}$, $A_{3}$, $A_{4}$, $A_{5}$), whose six subcodes are all distinct, and the t-code of $z_{s}$, $(A_{0}^\prime$, $A_{1}^\prime$, $A_{2}^\prime$, $A_{3}^\prime$, $A_{4}^\prime$, $A_{5}^\prime)$, is a permutation of the t-code of $z_{a}$. Since the permutation is nontrivial, at least two locations, say $l_{i}$ and $l_{j}$, the t-codes of $z_{a}$ and $z_{s}$ are not equal, see the following equation, 
\begin{eqnarray}
&&z_{a}\left( 
\begin{array}{cc}
l_{i} & l_{j} \\ 
A_{i} & A_{j}%
\end{array}%
\right) \\
&&z_{s}\left( 
\begin{array}{cc}
l_{i} & l_{j} \\ 
A_{i}^{\prime } & A_{j}^{\prime }%
\end{array}%
\right).  \notag
\end{eqnarray}%
That is, at location $l_{i}$, $A_{i}\neq A_{i}^{\prime }$, and at location $l_{j}$, $A_{j}\neq A_{j}^{\prime }$. Because the bases of the original cells $c_{a0}$ and $c_{s0}$ are both $a$, it follows that $4a+A_{i}\neq 4a+A_{i}^{\prime }$ at $l_{i}$ and $4a+A_{j}\neq 4a+A_{j}^{\prime }$ at $l_{j}$; hence $z_{a}\neq z_{s}$.

\textbf{Case (3-2)}. The t-codes of $z_{a}$ and $z_{s}$ contain repeated subcodes. For example, in Fig. 3, the t-code $(5,2,5)$ contains two subcodes equal to 5, and the t-code $(2,2,10,10)$ contains two subcodes equal to 2 and two equal to 10. We mark these t-codes in Fig. 3 by pink and green shading. Because $z_{a}$ and $z_{s}$ contain repeated subcodes, they may possibly be equal if the permutation of $c_{a0}$ happens to occur only at the locations of those repeated subcodes. Below, we prove that if this is the case, $z_{a}$ and $z_{s}$ are not disjoint, that is, they share at least one component cell.

\textbf{Case (3-2-1)}. The repeated subcodes occur at location indices 0 and 1, and such cases are marked with pink shading in Fig. 3. For example, the t-codes $(3,3,6)$ and $(2,2,10,10)$ have this property. This case indicates that $z_{a}\times TL(0,1) = z_{s}$, and hence $c_{a0}\times TL(0,1) = c_{s0}$. Since the minimal transition starting from an original cell is always  $TE(0,1) = TL(0,1)$, the relation  $c_{a0}\times TL(0,1) = c_{s0}$ is equivalent to $c_{a0}\times TE(0,1) = c_{s0}$, indicating that $c_{s0}$ is the minimal-transition cell of $c_{a0}$. Similarly, the reverse relation also holds; namely, $c_{s0}\times TL(0,1) = c_{a0}$, indicating that $c_{a0}$ is the minimal-transition cell of $c_{s0}$. As a result, $z_{a}$ and $z_{s}$ are joint in that they share the two cells $c_{a0}$ and $c_{s0}$.  

\textbf{Case (3-2-2)}. The repeated subcodes do not occur at location indices 0 and 1, and such cases are marked with green shading in Fig. 3. For example, the t-codes $(6,1,6,11)$, $(5,2,5)$, $(1,3,10,10,17,19)$, and $(1,4,7,14,14)$ have this property. By checking Fig. 3 we can see that the first subcode of their l-codes is always the minimal transition $TL(0,1)$, and the second is always $TL(l_{i},l_{j})$, where $l_{i}$ and $l_{j}$ are the locations of the repeated subcodes in the cluster t-codes. For example, the two subcodes equal to 6 in the t-code $(6,1,6,11)$ are located at $(0,2)$, and the second subcode of its l-code is also $(0,2)$; the two subcodes equal to 14 in the t-code $(1,4,7,14,14)$ are located at $(3,4)$, and the second subcode of its l-code is also $(3,4)$.

Starting from the original cell $c_{a0}(0,1,\ldots,l_{i},\ldots,l_{j},\ldots)$, the first two transitions are: (1) $TE(0,1)=TL(0,1)$, then $c_{a0}$ connects to $c_{a1}(1,0,\ldots ,l_{i},\ldots,l_{j},\ldots )$; (2) $TL(l_{i},l_{j})$, after which $c_{a1}$ connects to $c_{a2}(1,0,\ldots,l_{j},\ldots,l_{i},\ldots )$. If the second transition is applied to $c_{a0}$ instead, it connects to the cell $c_{g1}(0,1,\ldots ,l_{j},\ldots,l_{i},\ldots )$ in the group \textbf{G11}.

Because $c_{s0}$ is obtained from $c_{a0}$ by a $TL(l_{i},l_{j})$ transition, $c_{s0}$ has the t-code $(0,1,\ldots ,l_{j},\ldots ,l_{i},\ldots )$, which is exactly the same as that of the cell $c_{a1}$ in group \textbf{G11} generated from $c_{a0}$; that is, $c_{s0}=c_{a1}$. If no \textbf{G11} transition occurs, the first transition of $c_{s0}$ must also be $TE(0,1)=TL(0,1)$; hence $c_{s0}$ connects to $c_{s1}(1,0,\ldots,l_{j},\ldots,l_{i},\ldots)$, which is exactly the same as that of the two-step transition cells $c_{a2}$ in group \textbf{G21}; that is, $c_{s1}=c_{a2}$. Therefore, either $c_{s0}=c_{a1}$ or $c_{s1}=c_{a2}$, and in either case, $z_{a}$ and $z_{s}$ must share at least one component cell.

Based on the analyses of the above \textbf{Cases (1)}, \textbf{(2)}, and \textbf{(3)}, we obtain the property that in standard 2D SCM model, if two 4-cc clusters are equal, they must be joint. Therefore if they are not joint, they must be not equal, namely, the chromatic codes of 4-ccd clusters are uniqueness. However, the proof is only for original cells of the form given in Eq. (\ref{eq:ca0cs0}). The original cells may contain different transition intervals and thus take the following forms
\begin{eqnarray}\label{eq:Newbase}
	&&c_{a0}\left( 
	\begin{array}{cccccc}
		l_{a0} & l_{a1} & l_{a2} & l_{a3} & l_{a4} & l_{a5} \\ 
		a & a+1 & b+a+2 & b+a+3 & e+a+4 & e+a+5
	\end{array}
	\right)  \\
	&&c_{s0}\left( 
	\begin{array}{cccccc}
		l_{s0} & l_{s1} & l_{s2} & l_{s3} & l_{s4} & l_{s5} \\ 
		s & s+1 & g+s+2 & g+s+3 & h+s+4 & h+s+5
	\end{array}
	\right)  \notag
\end{eqnarray}
\begin{lemma}
	Suppose $x$ is a subcode of a cell. After the  value transitions within a $m$-cc cluster, the allowed range for this subcode in the other component cells is $[x-m+1, x+m-1]$. 
\end{lemma}
\begin{proof}
	A $m$-cc cluster indicates that the value transitions occurred $m-1$ times among its component cells. Since the value transitions between two connected cells are either $TE(x,x+1)$ or $TE(x,x-1)$, if the $m-1$ value transitions are all $TE(x,x+1)$, then $x$ increases to $x+1$, $x+2$, $\ldots$, and then $x+m-1$; if the three value transitions are all $TE(x,x-1)$, then $x$ decreases to $x-1$, $x-2$, $\ldots$, and then $x-m+1$.
\end{proof}
\begin{property}
	Suppose $[x,x+1,\ldots,x+n]$ is a transition interval of the t-code of an original cell. Increasing or decreasing the base $x$, do not change the connectivity among the component cells of the generated cluster, under the condition that the increased or decreased subcodes in the transition interval do not equal to the subcodes in other transition intervals.
\end{property}
\begin{proof}
	Since the neighboring value transitions are either $TE(x,x+1)$ or $TE(x,x-1)$, which make their chromatic distances $d_{1}=|x - (x+1)|+|(x+1)-x|$ or $d_{2}=|x - (x-1)|+|(x-1)-x|$ are always 2. Because the base $x$ is canceled in the above two distances, increasing or decreasing $x$ do not change the distance value 2 and hence the connection between the neighboring cells.
\end{proof}
If increasing the base $x$ to $x+k$ $(k\ge1)$, the corresponding chromatic subcodes of the generated $4$-cc cluster at least increase 4k, while for the subcodes in other transition intervals, according to the Lemma 1, it at most increases 3. Because $4k > 3$ with $k\ge1$, therefore increasing bases will not lead the cluster chromatic subcodes in the new transition interval equal to those in other transition intervals without increasing their bases. In consequence, if 4-ccd clusters generated from the original cells with t-codes Eq. (\ref{eq:ca0cs0}) are with the unique chromatic codes, then those generated from the original cells with t-codes of Eq. (\ref{eq:Newbase}) are also unique, and this will give the final proof of \textbf{Theorem 1}. 

\section{Discussion and conclusion}

The entities discussed in this work reside in the Euclidean space $\mathbb{R} ^{n}$. Not only in the $\mathbb{R}^{2}$ plane but also in higher $n$-dimensional spaces, the chromatic codes of 4-ccd clusters retain their uniqueness. As the spatial dimension increases, additional hidden cells emerge, and their dimensionality increases accordingly; however, the connectivity properties of their chromatic codes and the connectivity patterns of the clusters remain unchanged.

The uniqueness of the chromatic codes of connected disjoint clusters implies spatial heterogeneity. Many spaces, such as the internal space of matter, geographical space, and ecological structural space, exhibit heterogeneity and complexity. In a chromatic space, the chromatic code of any subspace (cluster) is unique, indicating that the spatial tessellations (e.g., Voronoi diagrams) composed of these subspaces also have unique codes.

The necessary conditions for the uniqueness of connected disjoint clusters are ``connectivity'' and ``disjointness''. If these two conditions are not met, i.e., if the component cells are either disconnected or joint, their chromatic codes may be equal. From the perspective of spatial locality, the characteristics of a given space are determined by the ``proximity'' and ``distinctness'' of its constituent spatial units. If disconnected spaces are mixed together, or if different spaces share many common units, their heterogeneity will be reduced. 

The chromatic codes of connected disjoint clusters formed from more cells, for example, 5-ccd and 6-ccd clusters, also exhibit uniqueness according to computer-based verification; however, this property has not yet been mathematically proven. If methods similar to those used in this study were adopted, one would need to consider a vast number of possible neighboring cell configurations, which would make the generated $m$-cc cluster family considerably more complex than the 4-cc cluster family shown in Fig. 3; consequently, the proof process would become exceedingly cumbersome. This study is intended as a starting point, in the hope that interested readers can propose more concise approaches to ultimately proving the uniqueness of the chromatic codes of $m$-ccd clusters.

\end{document}